\documentclass[a4paper]{article}

\usepackage[pages=all, color=black, position={current page.south}, placement=bottom, scale=1, opacity=1, vshift=5mm]{background}

\SetBgContents{}      % copyright

\usepackage[margin=1in]{geometry} % full-width

\usepackage{amsmath}
\usepackage{amssymb}
\usepackage{amsthm}
\usepackage{marvosym}
\usepackage[scr=boondox]{mathalpha}

\usepackage{csquotes}
\usepackage[
  style=numeric,
  sorting=nyt,
  sortcites=true,
  maxnames=5,
  minnames=4
]{biblatex}
\usepackage{algorithm}
\usepackage{algpseudocode}

\usepackage{tikz}
\usetikzlibrary{automata,positioning}

\usepackage{amsmath}
\usepackage{amsthm}
\usepackage{amssymb}

\usepackage[utf8]{inputenc}
\usepackage{hyperref}
\hypersetup{
	unicode,
	pdfauthor={Venturi Riccardo},
	pdftitle={Simplicity and irreducibility in circular automata},
	pdfsubject={Simplicity and irreducibility in circular automata},
	pdfkeywords={article, template, simple},
	pdfproducer={LaTeX},
	pdfcreator={pdflatex}
}

\theoremstyle{theorem}
\newtheorem{teo}{Theorem}[section]
\newtheorem{lemma}[teo]{Lemma}
\newtheorem{cor}[teo]{Corollary}
\newtheorem{conj}[teo]{Conjecture}
\newtheorem{opbm}[teo]{Open Problem}
\newtheorem{prop}[teo]{Proposition}

\newtheorem*{teo*}{Theorem}

\theoremstyle{definition}
\newtheorem{defi}[teo]{Definition}

\theoremstyle{remark}
\newtheorem{oss}[teo]{Remark}
\newtheorem{ese}[teo]{Example}

\usepackage{graphicx, color}
\graphicspath{{fig/}}

\usepackage{algorithm, algpseudocode} % use algorithm and algorithmicx for typesetting algorithms
\usepackage{mathrsfs} % for \mathscr command

\usepackage{lipsum}

\newcommand{\A}{\mathcal{A}}
\newcommand{\M}{\text{M}}

\newcommand{\Syn}{\text{Syn}}

\newcommand{\matn}{\mathbb{M}_{n-1}(\mathbb{C})}

\title{Simplicity and irreducibility in circular automata}
\author{
 Riccardo Venturi%
\footnote{Universidade NOVA de Lisboa, Lisbon, Portugal.
\texttt{r.venturi@campus.fct.unl.pt}}
}

\begin{document}

\maketitle
	
\begin{abstract}
This paper investigates the conditions under which a given circular (synchronizing) DFA is \emph{simple} (sometimes referred to as \emph{primitive}) and when it is \emph{irreducible}. Our notion of irreducibility slightly differs from the classical one, since we are considering our monoid representations to be over $\mathbb{C}$ instead of $\mathbb{Q}$; nevertheless, several well-known results remain valid—for instance, the fact that every irreducible automaton is necessarily simple. We provide a complete characterization of simplicity in the circular case by means of the \emph{weak contracting property}. Furthermore, we establish necessary and sufficient conditions for a circular \emph{contracting automaton} (a stronger condition than the weakly contracting one) to be irreducible, and we present examples illustrating our results.
\end{abstract}


\section{Introduction}

A complete deterministic automaton is an action of a free monoid $\Sigma^*$ (where $\Sigma$ is a finite set) on a finite set $Q$. More concretely, it can be described as a tuple $\mathcal{A} = (Q, \Sigma, \delta)$, where $Q$ is a finite set of \textit{states}, $\Sigma$ is a finite \textit{alphabet}, and $\delta: Q \times \Sigma \to Q$ is called the \textit{transition function}. This function fully describes the action of the alphabet $\Sigma$ on the state set $Q$ and extends naturally to words in $\Sigma^*$: nevertheless, rather than the functional notation, we shall adopt an action notation, writing $q \cdot x$ in place of $\delta(q,x)$. This action extends naturally, first to $\Sigma^*$, and then to any subset $F \subseteq Q$, in the obvious way by setting for any $u\in\Sigma^*$ 
\[ 
F \cdot u = \{ q \cdot u \mid q \in F \}.
\]
Unless specified, $n$ will always indicate the number of states $|Q|$, and for any word $u\in\Sigma^*$ we define the \textit{rank of} $u$ as $\text{rk}(u):=|Q\cdot u|$. From a graphic point of view, an automaton can be viewed as a directed, edge-labelled graph in which each vertex has exactly one outgoing edge labelled by each $a \in \Sigma$. A complete deterministic automaton is here referred to simply as an \textbf{automaton}, or sometimes as a \textbf{DFA}, and is one of the main theoretical tools in theoretical computer science used to recognize languages when an initial state and a set of final states are specified. However, our focus is on automata from a purely algebraic perspective. We will always assume our automata to be \textit{synchronizing}: an automaton is called synchronizing (or \textbf{reset}) if there exists a word $w \in \Sigma^{*}$, referred to as a \emph{synchronizing} (or \emph{reset}) word, that maps all states to a single, identical state. Formally, for a synchronizing word $w$, we have $q \cdot w = q' \cdot w$ for all states $q,q' \in Q$: it is well known that the set of synchronizing words $\Syn(\A)$ is a two-sided ideal of $\Sigma^*$. 
Given an automaton $\A = (Q,\Sigma,\delta)$, we define a \emph{congruence} to be an equivalence relation $\sigma \subseteq Q \times Q$ which is compatible with the action, that is, for every $u \in \Sigma^*$ and $p,q \in Q$, if $p \,\sigma\, q$ then $(p \cdot u) \,\sigma\, (q \cdot u)$. Let $p,q\in Q$ with $p\neq q$: we can define the \emph{congruence generated by} $\{p,q\}$, denoted $\langle \{p,q\}\rangle$, as the smallest congruence containing $(p,q)$. Another way of constructing the congruence generated by two states is considering the action $(p,q)\cdot \Sigma^*$, and then closing by reflexivity and transitivity. 
It is immediate that the diagonal and universal relations on $Q$, namely $\Delta_\A$ and $\nabla_\A$, are congruences of $\A$: if the only congruences of the given automaton are the latter two, we say that the automaton is \textit{simple} (also known in literature as \textit{primitive}, see \cite{RystIrr}). Every automaton we are considering throughout this paper is simple. A notable infinite class of simple automata is given by the \textit{\v{C}ern\'y automata} (see \cite{Vo_Survey} for their definition and, for instance, \cite{AlRo} for a proof of their simplicity). In particular, the first part of this work is dedicated to generalize the proof given in \cite{AlRo} of the simplicity of the \v{C}ern\'y automata to obtain a charecterization of simplicity for the whole class of \textit{circular} automata. We move then to the study of the \textit{irreducibility} for circular automata: it is well know that the class of irreducible automata lies in the simple one (see \cite{RystIrr}). Our goal will be to provide necessary and sufficient conditions for a circular (simple) automaton to be irreducible. Although we do not provide a complete characterization, we are able to establish several strong conditions, whose consistency is supported by explicit examples. We also address an implicit question raised in \cite{RodVen} concerning the existence of simple but non-irreducible automata: Example~\ref{ese:simp non irr} presents an infinite family of such automata, thereby settling the question. \\
One of the main motivations for this work is the belief that studying the structure of irreducible automata may provide new insights toward the solution of the \textit{Černý conjecture} (see \cite{Vo_Survey} for a comprehensive overview of the problem). The conjecture states that, for any synchronizing automaton $\A$, there exists a reset word $u\in\Sigma^*$ of length $|u|\leq (n-1)^2$. As shown in \cite{RodVen}, the class of simple automata is one of the three known classes for which a solution in these particular cases would imply a general solution. We believe that this reduction can be further refined by extending the analysis from the simple case to the irreducible one, since every known \textit{extremal} automaton (i.e., an automaton whose shortest reset word exactly reaches the bound) is irreducible. This brings us to state the following conjecture:

\begin{conj}\label{conj: extr are irred}
    Let $\A$ be an extremal synchronizing automaton. Then $\A$ is irreducible. 
\end{conj}

This conjecture is motivated by a class of examples that we are able to provide—namely, Example~\ref{ese:simp non irr}. In this case, we present an infinite family of simple but non-irreducible automata that can be synchronized by words of linear length in $n$. \\
Another motivation arises from the context of representation theory for finite monoids (see \cite{Stein_Book}). The representations under consideration here are related to (Rees quotients of) submonoids of the full transformation monoid. This automata-theoretic perspective allows us to leverage relatively simple combinatorial structures (monoids) to investigate particular algebraic objects that are often more challenging to address within a purely algebraic setting. Another related algebraic topic that we are (implicitly) addressing here concerns the theory of primitive monoids (for a complete dissertation, see \cite{Steinb_monoid}). A transformation monoid $M$ acting on a set $X$ is primitive if its action leaves no non-trivial partition of $X$ invariant. In this work, the transformation monoids $M$ are precisely the transition monoids associated with finite automata, and their actions are defined by the transition function. The first part of this paper establishes a necessary and sufficient condition for the transition monoid of a circular automaton (basically a transition monoid that contains the full cycle permutation $(1, \ldots, n)$) to be primitive.

\section{Preliminaries}

We begin here with our basic notions: an automaton $\A=(Q,\Sigma,\delta)$ is said to be \textit{circular} if $\exists a\in\Sigma$ such that:
$$\forall p,q\in Q. \ \exists k\geq 0 \text{ such that } p \cdot a^k = q \ \text{ or } q\cdot a^k = p$$
the letter $a\in\Sigma$ mentioned above is said to be \textit{circulating letter}. A relevant class of automata we are considering here is given by the so-called \textit{weakly defective} automata: these are automata such that for every letter $x\in\Sigma$ we have $\text{rk}(x)\geq n-1$, and it is straightforward to see that a weakly defective automaton admits a word of rank $2$.
We now introduce the representation-theoretic tools used in this work; for a comprehensive treatment of the representation theory of finite monoids, see \cite{Stein_Book}. Given an automaton $\A$, we always have an epimorphism $\pi: \Sigma^* \rightarrow \M(\A)$, where $\M(\A)$ is the transition monoid associated to $\A$. For convenience, we sometimes drop $\pi$ and identify elements in $\Sigma^*$ with elements in $\M(\A)$ and vice versa; so for instance if we are in the context of $\M(\A)$, we will write $\Syn(\A)$ instead of $\pi(\Syn(\A))$. With this convention, we define $\A^*:=\M(\A)/\Syn(\A)$ where the quotient taken here is the \textit{Rees quotient} of $\M(\A)$ over $\Syn(\A)$ (see \cite{Howie}). Let us now introduce the following construction. Let $Q=\{q_1,\ldots,q_n\}$ be the set of states, and let $\mathbb{C}Q \cong \mathbb{C}^n$ denote the free $\mathbb{C}$-vector space generated by $Q$. Observe that we have an action of $\M(\A)$ over $\mathbb{C}Q$ given by:
$$v \cdot u = (v_1q_1+\ldots+v_nq_n)\cdot u = v_1q_1\cdot u+\ldots+v_nq_n\cdot u$$
where $v\in \mathbb{C}Q$ and  $u\in\M(\A)$. Let now $\omega:=q_1 +\ldots +q_n$: by fixing the base $\{q_i-q_1\}_{i=2}^n$ for the ortogonal subspace $\omega^\perp$, it is easy to check that one may restrict the action of $\M(\A)$ from $\mathbb{C}Q$ to its (proper) subspace $\omega^\perp\cong \mathbb{C}^{n-1}$, and one may also check that for any $v\in \omega^\perp$ we have $v\cdot u = 0 \iff v \in \Syn(\A)$, and this induces a representation
$$\rho: \M(\A)/\Syn(\A) \rightarrow \text{End}_\mathbb{C}(\omega^\perp)\cong \matn$$
and thus $\rho:\A^*\rightarrow \matn$; we will refer to such a representation as the \textit{synchronized representation}  of $\A$. We can now give the following:

\begin{defi}
Let $\A$ be a synchronizing automaton. We say that $\A$ is \textit{irreducible} (or \textit{$\mathbb{C}$-irreducible}, to emphasize the field under consideration) if the associated synchronized representation $\rho:\A^* \to \M_n(\mathbb{C})$ is irreducible.
\end{defi}

Note that this definition of irreducibility for synchronizing automata differs slightly from that given, for instance, in \cite{RystIrr}, where the vector spaces are considered over the field $\mathbb{Q}$. We refer to that notion as \textit{$\mathbb{Q}$-irreducibility}. Nevertheless, our notion is more general, and therefore, if $\A$ is $\mathbb{C}$-irreducible, it is necessarily $\mathbb{Q}$-irreducible. This allows us to invoke any result that assumes $\mathbb{Q}$-irreducibility. In particular, we provide an example (Example~\ref{ese: Q irr, C red}) of a synchronizing automaton that is $\mathbb{Q}$-irreducible but $\mathbb{C}$-reducible.One of the key results known for ($\mathbb{Q}$-)irreducible automata, which also extends to our setting, is the following:

\begin{teo}
    Let $\A$ be an irreducible automaton. Then $\A$ is simple.
\end{teo}
\begin{proof}
    See \cite[Theorem~4]{RystIrr}.
\end{proof}

This shows that the class of irreducible automata is contained within the class of simple automata. As mentioned above, we next investigate the conditions under which a circular automaton $\A$ is irreducible. One of our main tools in this analysis will be the theory of \emph{circulant matrices} (see \cite{CircMat}). Recall that a circulant matrix is a complex matrix of the form:
\[\left(\begin{matrix}
    c_0      & c_{n-1} & \cdots  & c_2     & c_1     \\
c_1      & c_0     & c_{n-1} &         & c_2     \\
\vdots   & c_1     & c_0     & \ddots  & \vdots  \\
c_{n-2}  &         & \ddots  & \ddots  & c_{n-1} \\
c_{n-1}  & c_{n-2} & \cdots  & c_1     & c_0     \\
\end{matrix}\right)\]
and for any given circular automaton $\A$ equipped with its synchronizing representation $\rho$, we can always construct a circulant matrix whose rank provides information about the subspaces of $\omega^\perp$ that are invariant under the action of $\A^*$.

\section{Simplicity in circular automata}

In this section we give a combinatorial characterization of simplicity for circular automata. Throughout the whole paper, unless specified, we we only deal with (synchronizing) circular automata. Given a circular DFA $\mathcal{A}$ with $a\in\Sigma$ its circulating letter, we can define its \textit{relative metric} as follows: given $p,q\in Q$ let
$$d(p,q):=\min\{k\geq0 \ | \ p\cdot a^k = q  \ \text{ or } \  q\cdot a^k = p\}$$
observe that, for any $p,q\in Q$ with $p\neq q$ we have $1\leq d(p,q)\leq n/2$.

\begin{defi}
    A circular DFA $\mathcal{A}=(Q,\Sigma,\delta)$ is said to be \textit{contracting} if it satisfies the \textit{contracting property}, namely if $\forall p,q\in Q$ with $d(p,q)>1$, $\exists u \in \Sigma^*$ such that:
    $$1\leq d(p\cdot u, q\cdot u) < d(p,q)$$
\end{defi}

Observe that, one may restate the contracting property as follows: for every $p,q\in Q$ such that $d(p,q)>1$, there exists $u\in\Sigma^*$ such that $d(p\cdot u,q\cdot u)=1$. The equivalence can be proven by iterating the contracting property on two states and their images.

\begin{oss}\label{oss: d = min}    
    Let $\A$ be a circular automaton with $Q=\{q_1,\ldots,q_n\}$, $a\in\Sigma$ circulating letter such that $q_i\cdot a = q_{i+1 \mod{n}}$. Then, if $q_i,q_j\in Q$ with $i<j$ observe that the following equalities hold:
    $$q_i\cdot a^{j-i-1} = q_{j-1}\neq q_j \ \text{ and } \ q_j\cdot a^{n-j+i-1} = q_{n+i-1\mod{n}} = q_{i-1} \neq q_i$$
    which proves that $d(q_i,q_j)=\min\{j-i,n-j+i\}$.
\end{oss}

The main class of contracting automata is constituted by the well-known Černý automata, as established by the following proposition:

\begin{prop}
    The Černý automaton $\mathcal{C}_n$ is contracting for every $n\in \mathbb{N}$.
\end{prop}

\begin{proof}
    Let $p,q,p',q'\in Q$ such that $d(p,q) > 1$, $ d(p',q') =  1$. We know that $\mathcal{C}_n$ is double transitive for every $n\in\mathbb{N}$ (see \cite[Theorem 3.44]{RystovPhD}), meaning that we have a word $u\in\Sigma^*$ such that $p\cdot u = p'$, $q\cdot u = q'$ giving $d(p\cdot u,q\cdot u) = 1$.
\end{proof}

\begin{defi}
    Let $\A$ be a circulating automaton. We say that a couple $(p,q)\in Q\times Q$ is \textit{weakly contracting} if:
    $$\gcd\{(\gcd(d(p\cdot u,q\cdot u),\ n))\ | \  u\in\Sigma^* \ \text{ and } \ p\cdot u \neq q\cdot u\} = 1.$$
    An automaton is said to be \textit{weakly contracting} if every couple $(p,q)\in Q\times Q$ with $p\neq q$ is weakly contracting.
\end{defi}

\begin{oss}
    Let $\A$ be a synchronizing, circular automaton such that $|Q|=n$ is prime. Then $\A$ is weakly contractive.
\end{oss}

\begin{ese}[Non-weakly contracting, synchronizing automaton]
    Let $\A$ be the following automaton:
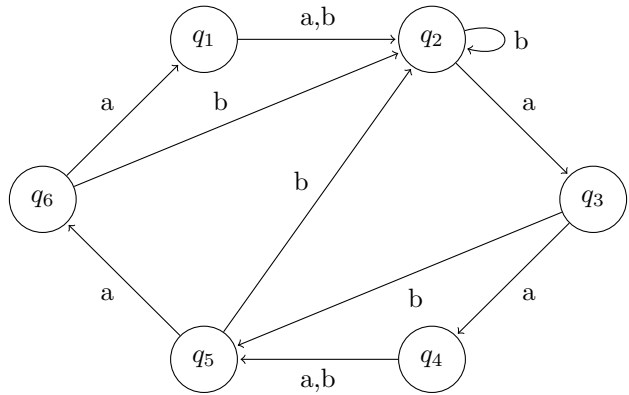
\begin{center}\begin{tikzpicture}[shorten >=1pt,node distance=3cm,on grid,auto] 
        \node[state] (q_1)   {$q_1$}; 
        \node[state] (q_2) [right=of q_1] {$q_2$}; 
        \node[state] (q_3) [below right =of q_2] {$q_3$}; 
        \node[state] (q_4) [below left=of q_3] {$q_4$};
        \node[state] (q_5) [left=of q_4] {$q_5$}; 
        \node[state] (q_6) [above left=of q_5] {$q_6$}; 
        \path[->] 
            (q_1) edge node {a,b} (q_2)
            (q_2) edge node {a} (q_3)
                  edge [loop right] node {b} ()
            (q_3) edge node {a} (q_4)
                  edge node {b}  (q_5)
            (q_4) edge node {a,b} (q_5)
            (q_5) edge node {a} (q_6)
                  edge node {b} (q_2)
            (q_6) edge node {a} (q_1)
                  edge node {b} (q_2);
\end{tikzpicture}\end{center}
and observe that it is synchronizing by means of the word $b^2$. Consider the couple $(q_1,q_4)$: it follows from an easy check that:
$$\text{for every }\ u\in\Sigma^* \ \text{ either } \ q_1\cdot u = q_4\cdot u \ \text{ or } \ d(q_1\cdot u,q_4\cdot u ) =3$$
showing us that $(q_1,q_4)$ is not weakly contractive.
\end{ese}

\begin{lemma}
    Let $\A$ be a contracting automaton. Then $\A$ is weakly contractive.
\end{lemma}

\begin{proof}
    Let $(p,q)\in Q\times Q$ with $p\neq q$. We have that $\exists u\in\Sigma^*$ such that $d(p\cdot u, q\cdot u) = 1$, and thus:
    $$\gcd(d(p\cdot u, q\cdot u),n) = \gcd(1,n) = 1$$
    which shows that any couple has to be weakly contracting.
\end{proof}

In general, to produce examples of synchronizing, weakly-contractive but not contractive automata, it is sufficient to consider a circular automaton with a prime number of states on a binary alphabet $\Sigma=\{a,b\}$ where $a$ is the circulating letter and $b$ is synchronizing. One may think that this class of examples represents the only possibile weakly-contractive non-contractive automata. However, the following example disproves this claim.

\begin{ese}[Weakly contractive, non-contractive automaton]\label{ese: weakly-con non-con}
Let $\A$ be the following automaton:
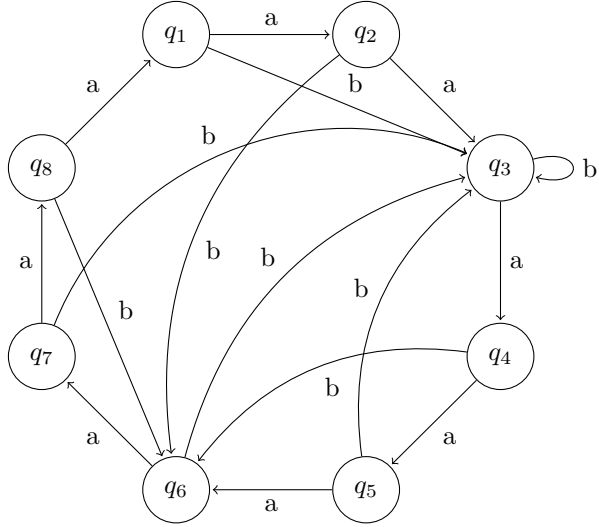
\begin{center}\begin{tikzpicture}[shorten >=1pt,node distance=2.5cm,on grid,auto] 
        \node[state] (q_1)   {$q_1$}; 
        \node[state] (q_2) [right=of q_1] {$q_2$}; 
        \node[state] (q_3) [below right =of q_2] {$q_3$}; 
        \node[state] (q_4) [below=of q_3] {$q_4$};
        \node[state] (q_5) [below left=of q_4] {$q_5$}; 
        \node[state] (q_6) [left=of q_5] {$q_6$}; 
        \node[state] (q_7) [above left =of q_6] {$q_7$}; 
        \node[state] (q_8) [above=of q_7] {$q_8$};
        \path[->] 
            (q_1) edge node {a} (q_2)
                  edge node {b}  (q_3)
            (q_2) edge node {a} (q_3)
                  edge [bend right] node {b} (q_6)
            (q_3) edge node {a} (q_4)
                  edge [loop right] node {b}  ()
            (q_4) edge node {a} (q_5)
                  edge [bend right] node {b}  (q_6)
            (q_5) edge node {a} (q_6)
                  edge [bend left] node {b}  (q_3)
            (q_6) edge node {a} (q_7)
                  edge [bend left] node {b} (q_3)
            (q_7) edge node {a} (q_8)
                  edge [bend left=46] node {b} (q_3)
            (q_8) edge node {a} (q_1)
                  edge node {b} (q_6);
\end{tikzpicture}\end{center}
Let $(p,q)\in Q\times Q$ such that $p\neq q$. If $d(p,q)=1,3$, we necessarily have that $(p,q)$ is weakly contractive. Then we have two cases:
\begin{itemize}
    \item if $d(p,q)=2$, then $\exists k\geq 0$ such that $\{p\cdot a^k,q\cdot a^k\} = \{q_4,q_6\}$ and thus we have that $d(p\cdot a^kb,q\cdot a^kb)= 3$, showing that every couple at distance $2$ has to be weakly contracting.
    \item if $d(p,q)=4$, then again $\exists k\geq 0$ such that $\{p\cdot a^k,q\cdot a^k\} = \{q_2,q_6\}$ and thus we have that $d(p\cdot a^kb,q\cdot a^kb)= 3$, showing that every couple at distance $4$ has to be weakly contracting.
\end{itemize}
To see that $\A$ is not contracting, it is enough to check that every couple at distance $\geq2$ is never sent in two consecutive states, which follows from an easy computation.
\end{ese}

\begin{lemma}\label{lemma: gcd n e distanza}
    Let $\A$ be a circular DFA and $p,q\in Q$ such that $p\neq q$. Then if $\gcd(d(p,q),n) = k$ we have that $\exists (p_1,q_1)\in\langle\{p,q\}\rangle$ such that $d(p_1,q_1)=k$. In particular, if $\gcd(d(p,q),n) = 1$ we have that $\langle\{p,q\}\rangle = \nabla_\A$.
\end{lemma}

\begin{proof}
    If $d(p,q)=k$, we are done. Otherwise, let us write $d=d(p,q)$ and assume without loss of generality that $p\cdot a^d = q$. Observe that $(p\cdot a^d,q\cdot a^d) = (q,q\cdot a^d)\in\sigma\Rightarrow(p,q\cdot a^d)\in\sigma$: iterating this process, we have that for any $s\geq 0$ we have $(p,q\cdot a^{sd})\in\sigma$. Being $\gcd(d,n)=k$, we have that $\exists r\geq 1$ such that $rd = k \mod{n}$: thus we can write:
    $$(p,q\cdot a^{rd}) = (p,q\cdot a^{tn + k}) = (p,q\cdot a^k)\in\sigma \Rightarrow (q,q\cdot a^k)\in\sigma$$
    which concludes since $k<d \leq n/2$ and thus $d(q,q\cdot a^k) = k$.
\end{proof}

\begin{oss}
    The previous lemma shows that if $n$ is prime, we must have that $\A$ is simple: indeed for any $p,q\in Q$ we would have $\gcd(d(p,q),n)=1$.
\end{oss}

Let us move to the proof of the following (technical) lemma. In what follows, we define $a^{-i}:=a^{n-i}$ if $0\leq i\leq n$.

\begin{lemma}\label{lemma: gcd 2 distanze}
    Let $\A$ be a circular automaton and $\sigma$ a congruence of $\A$. Let $ (p_1,q_1),(p_2,q_2)\in \sigma$ such that:
    $$\gcd(d(p_1,q_1),d(p_2,q_2)) = k.$$
    Then $\exists (p,q)\in\sigma$ such that $d(p,q)=k$.
\end{lemma}

\begin{proof}
    If $d(p_1,q_1)=d(p_2,q_2)$, we are done. Let then $d_1=d(p_1,q_1), \ d_2=d(p_2,q_2)$ and assume (up to applying $\cdot a^k$, for some $k\geq 0)$ that $p_1=p_2$ with $d_1< d_2$: we have that $\exists s,t\in\mathbb{Z}$ such that $sd_1+td_2 = k$. Observe now that $(p_1,p_1\cdot a^{md_1}), (p_2,p_2\cdot a^{md_2})\in\sigma$ for every $m\in\mathbb{Z}$. Observe that $(p_1,p_1\cdot a^{sd_1}),(p_1,p_1\cdot a^{-td_2})\in\sigma$ and thus $(p_1\cdot a^{sd_1},p_1\cdot a^{-td_2})\in\sigma$. Now, we have
    $$d(p_1\cdot a^{sd_1},p_1\cdot a^{-td_2}) = (\min\{|sd_1+td_2|, n- |sd_1+td_2|\}) = |sd_1+td_2| = k \mod{n}$$
    since $k<d_1,d_2$. The statement is now satisfied by taking $(p,q) = (p_1\cdot a^{sd_1},p_1\cdot a^{-td_2})\in\sigma$.
\end{proof}

The following proposition provides a generalization to \cite[Proposition 7]{AlRo}:

\begin{prop}\label{prop: w.c. implies simple}
    A weakly contracting DFA $\mathcal{A}$ is simple. In particular, the Černý automaton $\mathcal{C}_n$ is simple for every $n\in \mathbb{N}.$
\end{prop}

\begin{proof}
    Assume $n\geq 3$ (otherwise it is obvious) and let $p,q \in Q$ such that $d(p,q)>1$: by the weak contracting property, $\exists u_1,\ldots,u_k\in\Sigma^*$ such that 
    $$\gcd(k_1,\ldots,k_n) = 1 \ \text{ where } \ k_i:=\gcd(d(p\cdot u_i,q\cdot u_i),n).$$
    Let $\sigma=\langle(p,q)\rangle$. Observe that, by means of Lemma \ref{lemma: gcd n e distanza}, we must have that for every $i\in\{1,\ldots,k\}$ there exists $(p_i,q_i)\in\sigma$ such that $d(p_i,q_i)=k_i$. Iteratively applying Lemma \ref{lemma: gcd 2 distanze}, at most $k$-times, we obtain a pair $(t,s)\in\sigma$ such that $d(t,s)=1$. Hence, we easily conclude that  $\sigma=\nabla_\A$.
\end{proof}

\begin{lemma}\label{lemma: divisioni}
    Let $\A$ be a circular automaton, $\sigma$ a congruence on $\A$ and $(p,q),(q,t)\in\sigma$. Then if $m\geq1$ is such that $m$ divides $d(p,q),d(q,t)$ and $n$, then $m$ divides $d(p,t)$.
\end{lemma}

\begin{proof}
    Let us compute $d(p,t)$. Let $Q=\{q_1,\ldots,q_n\}$ and assume $p=q_i, \ q=q_j,\ t=q_k$ for $i,j,k\in\{1,\ldots,n\}$. Observe that we have the following: 
    \[ \begin{split}
        d(p,q) & = \min\{|i-j|, \ n-|i-j|\} \\
        d(q,t) & = \min\{|j-k|,\ n-|j-k|\}
    \end{split}\]
    and since $m$ divides both $d(p,q)$ and $d(q,t)$, in any case we must have that $m$ divides $|i-j|$ and $|j-k|$. Observe now that: 
    $$d(p,t)=\min\{|i-k|, \ n-|i-k|\}$$
    and $|i-k| = |(i-j) + (j-k)|$. We clearly have then that $m$ divides $|i-k|$ and thus, in any case, $m$ divides $d(p,t)$.
\end{proof}

\begin{prop}\label{prop: simp implies w.c.}
    Let $\mathcal{A}$ be a simple circular DFA. Then $\mathcal{A}$ is weakly contractive.
\end{prop}

\begin{proof}
    Assume that $\A$ is not weakly contractive, let $(p,q)\in Q\times Q$ be such that:
    $$k:=\gcd\{\gcd(d(p\cdot u,q\cdot u), \ n) \ | \ u\in\Sigma^* \ \text{ and } \ p\cdot u \neq q\cdot u\} > 1$$
    and let $\sigma=\langle(p,q)\rangle$. We claim that for any $(t,s)\in\sigma$, we must have $k|d(t,s)$: 
    indeed let $(p_1,p_k)\in\sigma$ such that $p_1\neq p_k$, we can write
    $$p_1\sigma p_2\sigma\ldots\sigma p_k \ \text{ where } \ \{p_i,p_{i+1}\} = \{p,q\}\cdot u \ \text{ for some }\ u\in\Sigma^*$$
    and $m | d(p_i,p_{i+1})$ for any $i\in\{1,\ldots,k-1\}$. By means of Lemma \ref{lemma: divisioni}, we must have $m | d(p_i,p_{i+1})$, $ m | d(p_{i+1},p_{i+2}) \Rightarrow m|d(p_i,p_{i+1})$ which iteratively shows that $m|d(p_1,p_k)$. Being $m>1$ this fact implies that there are no couples in $\sigma$ having distance one: this shows that $\sigma$  has to be a proper congruence, which is a contradiction since we have $\A$ simple by hypothesis.
\end{proof}

By combining Proposition \ref{prop: simp implies w.c.} and Proposition \ref{prop: w.c. implies simple}, we easily get the following characterization. 

\begin{cor}\label{cor: caratt circ simple}
    Let $\mathcal{A}$ be a circular DFA. Then $\mathcal{A}$ is simple $\iff \mathcal{A}$ is weakly contractive.
\end{cor}

\section{Irreducibility in circular automata}

This last section is dedicated to study necessary and sufficient conditions for a (circular) automaton to be irreducible: from now on we will assume $\A$ to be circular and synchronizing. Let $\A$ be an automaton with $a\in\Sigma$ its circulating letter. It is not difficult to see that the matrix associated to $a$ with respect to the base $\{q_i-q_1\}_{i=2}^n$ of $\omega^\perp$, in $\text{End}_\mathbb{C}(\omega^\perp) \cong\matn$ has the following form:
\[a\mapsto A:=\left(\begin{matrix}
    -1 & -1 & \ldots  & \ldots & -1\\
    1  & 0 & \ldots   & \ldots & 0\\
    0 & 1  & 0 & \ldots   & 0 \\
    &&\ldots&&&\\
    0 & \ldots & 0 & 1 & 0
\end{matrix}\right)\]
in other words, a matrix having all $-1$ on the first row and the identity matrix $(n-2)\times(n-2)$ as a submatrix (when we remove the first row and last column). We begin by giving the following construction: let $V\subseteq \omega^\perp$ such that $V\neq\{0\}$ and $V$ is invariant under the action of $\A^*$. Let $A\in\matn$ as above: then if $v\in V$, observe that $A^{k}v\in V$ for every $k\in\{1,\ldots,n-1\}$. Suppose in the local coordinates of $\omega^{\perp}$ we have $v=(v_1,\ldots,v_{n-1})$ and consider the following vector of $\mathbb{C}^n$
$$\underline{w}=(v_1,\ldots,v_{n-1},-\sum_{i=1}^{n-1}v_i).$$
Observe that $\underline{w}$ is nothing but $v$ written in the base $\{q_1,\ldots,q_n\}$. Let now $C$ be the \textit{circulant matrix} associated to $\underline{w}$ (i.e. having $\underline{w}$ as first column: see \cite{CircMat}). Its associated polynomial is:
$$f_{v}(x) = v_1 + v_2x + \ldots + v_{n-1}x^{n-2} + (-\sum_{i=1}^{n-1}v_i)x^{n-1}$$
Note that $f_{v}(1)=0$ clearly holds. Now, by removing the $n$-th row from $C$, it is easy to see that we obtain an $(n-1)\times n$ matrix whose columns are $v, \ Av, \ A^2v, \ \ldots, \ A^{n-1}v$: let indeed $(x_1,\ldots,x_{n-1})\in\mathbb{C}^{n-1}$ and define $x_n := -\sum_ix_i$. Observe that:
\[ A \left(\begin{matrix}
    x_1\\x_2\\x_3 \\\ldots\\x_{n-1}
\end{matrix}\right) = \left(\begin{matrix}
    -1 & -1 & \ldots  & \ldots & -1\\
    1  & 0 & \ldots   & \ldots & 0\\
    0 & 1  & 0 & \ldots   & 0 \\
    &&\ldots&&&\\
    0 & \ldots & 0 & 1 & 0
\end{matrix}\right)\left(\begin{matrix}
    x_1\\x_2\\x_3 \\\ldots\\x_{n-1}
\end{matrix}\right) = \left(\begin{matrix}
    -\sum_ix_i\\x_1\\x_2 \\\ldots\\x_{n-2}
\end{matrix}\right) = \left(\begin{matrix}
    x_n\\x_1\\x_2 \\\ldots\\x_{n-2}
\end{matrix}\right) \]
and by induction we obtain the required "cyclicity" of $A^kv$ for any $k$. This suggests that studying the rank $\operatorname{rk}(C)$ can provide valuable information about whether our candidate vectors actually generate the whole space $\omega^{\perp}\simeq\mathbb{C}^{n-1}$. A well know result (see \cite{CircMat}) states that, if $d=\deg(\gcd(f(x),x^n-1))$, we have:
$$\text{rk}(C) = n- d$$
and since $f(1)=0$, we necessarily have $d\geq 1$. Our interest is to understand whether the other (non-trivial) roots of unity are zeroes for $f(x)$ as well.

\begin{lemma}\label{lemma: span gcd}
    In the above setting, assume that  $\deg(\gcd(f_v(x),x^n-1)) = 1$ for some $v\in \mathbb{C}^{n-1}\setminus\{0\}$. Then we have that:
    $$\langle v, Av, A^2v,   \ldots  ,   A^{n-1}v\rangle = \mathbb{C}^{n-1}$$
    in other words, we must have $V=\mathbb{C}^{n-1}$.
\end{lemma}

\begin{proof}
    Let $C$ be the circulant matrix as defined above. Observe that here we have $\text{rk}(C) = n- d = n-1$ which means that $C$ admits a non-zero $(n-1)-$minor. Let $M_{i,j}$ be that minor: it is well known that (see \cite{CircMat}) $$ M_{i,j} = M_{i+k \ \text{mod} (n) , \ j+k\ \text{mod} (n)} \ \text{ for any }\ k \geq 0$$
    and thus we have that $M_{i,j}\neq 0 \Rightarrow M_{n,j+n-i \ \text{mod}(n)} \neq 0$. By construction we know that removing the last row from $C$ we obtain a matrix having as columns the vectors $v,Av,\ldots,A^{n-1}v$: the condition on the minor gives that we must have $(n-1)-$linearly independent vectors among them, and this concludes the proof of the lemma. 
\end{proof}

%\begin{proof}
 %   Let $C$ be the circulant matrix as defined above and observe that its rows are linearly dependent: indeed, their sum is equal to zero due to the equality $v_1+v_2+\ldots + v_n = 0$.
%\end{proof}

We now turn to the case where the action of $\A$ leaves invariant a one-dimensional subspace $V$, namely $V = \langle v \rangle$ with $v \neq 0$. In this situation we must have $Av \in \langle v \rangle$, which means that $Av = kv$ for some scalar $k \neq 0$. Hence, $v$ is an eigenvector of $A$, which motivates the following lemma:

\begin{lemma}\label{lemma: eigen A}
    Let $\A$ be a synchronizing circular automaton and $A$ the matrix associated to the circulating letter. Then we have that the $n-1$ distinct eigenvalues of $A$ are:
    $$\text{Spec}(A) = \{e^{\frac{2\pi ik}{n}} \ |  \ k\in \{1,\ldots,n-1\}\}$$
    and the eigenvectors are given by: 
    $$(e^\frac{2(n-1)k\pi i}{n},e^\frac{2(n-2)k\pi i}{n},\ldots,e^\frac{2k\pi i}{n})\in\mathbb{C}^{n-1}$$
    for $k\in\{1,\ldots,n-1\}$.
\end{lemma}

\begin{proof}
    Let us start by computing $p(\lambda)=\det(A-\lambda I_{n-1})$ by induction on $n$: out goal will be to prove that
    $$p(\lambda)= (-1)^{n-1}(\sum_{i=0}^{n-1}\lambda^i)$$
  We know that, for any $n$, the matrix associated to the circulating letter is always given by a first row only containing $-1$, an identity matrix $I_{n-2}$ as submatrix obtained removing the first row and last column, and everything else is equal to $0$. Observe that, by removing the last row and last column, we obtain the matrix associated to the circulating letter for a $(n-1)$-state automaton: this suggests us to work by induction on $n$. The base case is trivial: if $n=2$ we have that $A=(-1)$, and thus we obtain $p(\lambda)=-\lambda-1$ as required. Assume the induction hypothesis: by means of the Laplace formula together with the fact that the last column of $A-\lambda I_{n-1}$ has $-1$ in position $(1,n-1)$, and $-\lambda$ in position $(n-1,n-1)$ and $0$ otherwise, we have that:
    \[\begin{split} p(\lambda) & =\det(A-\lambda I_{n-1}) \\ & =  (-\lambda)\cdot((-1)^{n-2}(\sum_{i=0}^{n-2}\lambda^i)) + (-1)^{n+1} \\ &  =
     (-1)^{n-1} + (-1)^{n-1}(\sum_{i=1}^{n-1}\lambda^i) \\
     & =(-1)^{n-1}(\sum_{i=0}^{n-1}\lambda^i)\end{split} \]
    thanks to the fact that the matrix $A_{1,n}$ obtained by removing the first row and last column is diagonal with the diagonal equal to $(1,1,\ldots,1)$. This concludes the first part of the proof. Observe now that, for a given non-zero eigenvector $v=(v_1,\ldots,v_{n-1})$ of $A$ we have:
    \[Av = A\left(\begin{matrix}
     v_1 \\ v_2 \\ \ldots \\ v_{n-1}
    \end{matrix}\right) =  \left(\begin{matrix}
     -\sum_i v_i \\ v_1 \\ \ldots \\ v_{n-2}
    \end{matrix}\right) = e^\frac{k2\pi i}{n}\left(\begin{matrix}
     v_1 \\ v_2 \\ \ldots \\ v_{n-1}
    \end{matrix}\right)
    \]  
    and it is straightforward to check that, with this constraint:
    $$v_1=e^\frac{2k\pi i}{n} v_2 = e^\frac{4k\pi i}{n} v_3 = \ldots = e^\frac{2(n-2)k\pi i}{n} v_{n-1}$$
    Observe that we have $v_{n-1}=0\Rightarrow v=(0,\ldots,0)$. Assume then $v_{n-1}\neq 0$ and up to a scalar multiplication, we have $v_{n-1}= e^\frac{k2\pi i}{n}$: by the given relation between the $v_i$'s, we conclude.
\end{proof}

\begin{cor}\label{cor: coord 0, inv di dim 1}
    Let $\A$ be a circular automaton such that there exists a one-dimensional subspace  $\langle w\rangle \subseteq \mathbb{C}^{n-1}$ invariant under the action of $\A^*$. Then if $w\cdot u$ has has a zero local coordinate for some $u\in \Sigma^*$, we have $w\cdot u =0$.
\end{cor}

Thanks to the previous result, we are able to give the following example, showing that our notion of irreducibility for synchronizing automata differs strictly from the one given, for instance, in \cite{RystIrr}.

\begin{ese}[$\mathbb{Q}$-irreducible, $\mathbb{C}$-reducible automaton]\label{ese: Q irr, C red}
Let $\A$ be the following automaton:
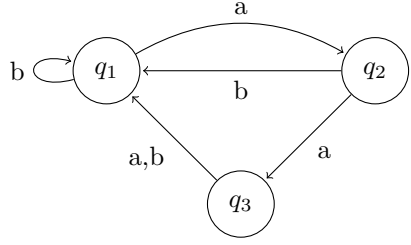
\begin{center}\begin{tikzpicture}[shorten >=1pt,node distance=2.5cm,on grid,auto] 
        \node[state] (q_3) {$q_3$};       
        \node[state] (q_1) [above left=of q_3]  {$q_1$}; 
        \node[state] (q_2) [above right=of q_3] {$q_2$}; 
        \path[->] 
            (q_1) edge [bend left] node {a} (q_2)
                  edge [loop left] node {b}  ()
            (q_2) edge node {a} (q_3)
                  edge node {b} (q_1)
            (q_3) edge node {a,b} (q_1);
\end{tikzpicture}\end{center}
Observe that $\A^*$ has to be simple since $n=|Q|$ is prime. Let now $\{0\}\neq V\subsetneq \mathbb{Q}^2$ be such that $V$ is invariant under the action of $\A^*$: we must have that $V=\langle v\rangle$ with $v\in\mathbb{Q}^2$ eigenvector of $A$, the matrix associated with the circulating letter. By means of the previous lemma we have that the associated eigenvalues are roots of unity. Note that all the $3-$rd roots of unity are complex, and thus there is no such $v$ in $\mathbb{Q}^2$: this gives us that $\A$ has to be $\mathbb{Q}-$irreducible. Observe now that the space generated by the vector $v = (e^\frac{2\pi i}{3},e^\frac{4\pi i}{3})$ is invariant under the action of $\A^*$ if considered over $\mathbb{C}$, showing us that $\A$ has to be $\mathbb{C}-$reducible.
\end{ese}

We move now to the proof of a theorem providing strong conditions on the rank of non-group elements, in the case where the representation admits a one-dimensional invariant subspace. 
Let us begin by proving the following (trivial) lemma:

\begin{prop}\label{prop: 0 in pos j}
    Let $\A$ be a synchronizing automaton (not necessarily circular) with $Q=\{q_1,\ldots,q_n\}$, $v\in\omega^\perp\cong \mathbb{C}^{n-1}$ and $j\in\{2,\ldots,n\}$. Then $v_j$ has a zero coefficient in the base $Q$ at position $j$ if and only if it has a zero in the base $\{q_i-q_1\}_{i=2}^n$ at  position $j-1$. 
\end{prop}

\begin{proof}
    The statement follows directly by expliciting $v$ in both the bases. Observe that:
    \[
    \begin{split}
        v & = v_2(q_2-q_1) + \ldots + v_{j-1}(q_{j-1}-q_1)+v_{j+1}(q_{j+1}-q_1)+\ldots+v_n(q_n
-q_1) \  \text{ in the base } \ \{q_i-q_1\}_{i=2}^n \\
    v & = (-\sum_iv_i)q_1 + v_2q_2 + \ldots + v_{j-1}q_{j-1}+v_{j+1}q_{j+1} +\ldots + v_{n}q_n \  \text{ in the base }\ Q \ \text{ of } \ \mathbb{C}Q
    \end{split}
    \]
    and this clearly concludes.
\end{proof}

\begin{teo}\label{teo: bound red 1dim}
    Let $\A$ be a circular automaton such that the action of $\A^*$ is invariant on a one dimensional subspace $\langle v\rangle\subsetneq \mathbb{C}^{n-1}$ with $v\neq 0$. Then, for every $u\in\Sigma^*$ such that $\text{rk}(u)\leq n-1$, we must have $\text{rk}(u)\leq n/2$.
\end{teo}
\begin{proof}
    Being the action of $\A^*$ invariant on $\langle v\rangle$, as mentioned above we must have that $v\cdot a^i\in\langle v\rangle$ for every $i$: this is equivalent to saying that $Av = kv$ and thus $v$ must be an eigenvector of $A$. By means Lemma \ref{lemma: eigen A}, if  $v=(v_2,\ldots, v_n)$ we must have $v_i\neq 0$ for every $i$. Let $u\in\Sigma^*$ be such that $\text{df}(u)\geq 1$: we claim that we can find a vector in $\langle v\rangle$  having a coefficient equal to $0$ with respect to the base $\{q_i-q_1\}_{i=2}^n$. Indeed, we have that $\text{rk}(u)< n$, and if $q_1\in Q\cdot u \Rightarrow \exists j\neq 1$ such that $q_j \notin Q\cdot u$, implying that $v\cdot u$ has a $0$ in the $j-$th position with respect to the base $Q$ of $\mathbb{C}Q$. By means of Proposition \ref{prop: 0 in pos j}, we deduce that $v$ admits a zero in the $(j-1)-$th position as element of $\omega^\perp$. Otherwise, if $Q \cdot u = Q\setminus \{q_1\}$, observe that $Q\cdot ua = Q\setminus \{q_2\}$ and thus we can conclude as above by substituting $u$ with $ua$. This gives us that $v\cdot u$ has at least one local coordinate equal to 0, and by Corollary \ref{cor: coord 0, inv di dim 1} this forces $v\cdot u = 0$.
    We claim that for every $q_i$ there is a $j\neq i$ such that $q_i\cdot u= q_j\cdot u$. Assume by contradiction that  $\exists i \in\{1,\ldots,n\}$ such that $q_i\cdot u \neq q_j\cdot u$ for all $j\neq i$. If $i=1$, again we substitute $u$  with $au$ and apply the following argument. By definition of $v$ we have that every component of $v$ is non-zero, and observe that we have the two following cases:
    \begin{itemize}
        \item if $q_i\cdot u = q_1$, we have that in the base $Q$ of $\mathbb{C}Q$ of $\mathbb{C}^n$ $v\cdot u= v^*= (-\sum_jv_j^*)q_1+\ldots+v_n^*q_n$, and we must have  $-\sum_jv_j^* = v_i \neq 0$, showing that at least one of the last $n-1$ coefficients in the base $Q$ of $v\cdot u$ has to be non zero, and again by Proposition \ref{prop: 0 in pos j} it is a contradiction.
        \item if $q_i\cdot u = q_j$ with $j\neq 1$, we have that $v\cdot u$ has a non zero coefficient in position $j$ if written in the base $Q$, which is exactly equal to $v_i$. Again, we can conclude as before by means of Proposition \ref{prop: 0 in pos j}.
    \end{itemize}
    Therefore, we conclude that for any fixed $u\in\Sigma^* $ such that $\text{df}(u)\geq 1$ we have that $\forall i$, $\exists j\neq i$ such that $p_i\cdot u = p_j\cdot u$, which implies $\text{rk}(u)\leq n/2$. 
\end{proof}

The following example provides a class of synchronizing automata proving that the bound given in the previous theorem has to be sharp:

\begin{ese}[Non-irreducible, simple automata]\label{ese:simp non irr}

Let $n\in\mathbb{N}$ be an even number and consider the automaton: $\A = (Q,\Sigma,\delta)$ where $Q=\{q_1,\ldots,q_{n}\}$, $\Sigma=\{a,b\}$ with $a$ circulating letter such that $q_i\cdot a = q_{i+1\text{ mod }n}$ and 

\[
q_i\cdot b = \begin{cases}
    q_i & \text{if } \ i\leq n/2 \\
    q_{n-i+1} & \text{if } \ i>n/2 \\
\end{cases}
\]
and observe that $(ba)^{n/2}b\in\Syn(\A)$. It is straightforward to check that the action of $\A^*$ is invariant on the subspace of $\mathbb{C}^{n-1}$ generated by $v=(-1,1,-1,\ldots,-1)$: indeed, we have that $v\cdot a = -v$ and $v\cdot b=0$. This shows that $\A$ has to be reducible. Observe that $\A$ is simple: let $q_i,q_j\in Q$ with $d(q_i,q_j)> 1$ and $i<j$, and assume without loss of generality that $i,j\leq n/2$ (we can always apply a suitable power of $\cdot  a$ to obtain this hypothesis). Let $k\geq 1$ be the minimum such that $q_j\cdot a^k = q_{n/2+1}$, and observe that:
$$1\leq d(q_i\cdot a^kb,q_j\cdot a^kb) = d(q_i\cdot a^k,q_{j}\cdot a^{k-1})<d(q_i,q_j)$$
showing that $\A$ has to be contracting and thus simple. Note also that for any $x\in\Sigma$ either $\text{rk}(x) = n$ or $\text{rk}(x)=n/2$, which gives us that the bound given in the statement of Theorem \ref{teo: bound red 1dim} has to be sharp.  The following graph depicts the $n=4$ case:
\begin{center}\begin{tikzpicture}[shorten >=1pt,node distance=3cm,on grid,auto] 
        \node[state] (q_1)   {$q_1$}; 
        \node[state] (q_2) [right=of q_1] {$q_2$}; 
        \node[state] (q_3) [below=of q_2] {$q_3$}; 
        \node[state] (q_4) [below=of q_1] {$q_4$};
        \path[->] 
            (q_1) edge [bend left] node {a} (q_2)
                  edge [loop left] node {b}  ()
            (q_2) edge node {a} (q_3)
                  edge [bend left] node {b} (q_1)
            (q_3) edge node {a,b} (q_4)
            (q_4) edge node {a} (q_1)
                  edge [loop left] node {b}  ();
\end{tikzpicture}\end{center}
\end{ese}

We now turn to studying sufficient conditions for an automaton to be irreducible. We begin with a well-known result of Rystov (\cite[Theorem~4]{RystIrr}). Although this theorem was originally proved for irreducibility over $\mathbb{Q}$, for the sake of completeness, we provide a proof here in the setting of $\mathbb{C}$.

\begin{teo}\label{teo: Rystov}
    Let $\A$ be a simple, weakly defective automaton. Then $\A$ is irreducible.
\end{teo}

\begin{proof}
    Let $v\in\omega^\perp\setminus\{0\}$ and let $V=\langle v\cdot \A^*\rangle$ be the subspace of $\omega^\perp$ generated by the action of $\A^*$. We clearly have that $V$ is invariant under the action of $\A^*$. Being $\A$ synchronizing, we must have a minimal-length word $u\in\Sigma^*$ such that $v\cdot u = 0$. Being $\A$ weakly defective, we can write $u=yx$ with $x$ letter of defectivity one.  Now, assume that $q_i\cdot x = q_j\cdot x$, hence the kernel of the matrix associated to $x$ can be written as:
    $$\ker(x) = \langle (q_i-q_1)-(q_j-q_1)\rangle.$$
    Observe that $v\cdot y \in\ker(x)$. Being $u$ minimal, we must have $v\cdot y \neq 0 \Rightarrow v\cdot y = k((q_i-q_1)-(q_j-q_1))$, with $k\in\mathbb{C}\setminus\{0\}$: this implies that $(q_i-q_1)-(q_j-q_1) \in V$. Being $\A$ simple, we must have $(q_s-q_1)-(q_k-q_1)\in V$ for every $s,k\in\{1,\ldots,n\}$: indeed, observe that if $\sigma = \langle\{q_i,q_j\}\rangle $, being $\A$ simple we must have that it is equal to $\nabla_\A$. This means that there exist a chain of indices in $\{i,j\}\cdot \A^*$ such that, if closed under transitivity, gives $\{s,k\}$. In other words, there exist $q_{i_1},\ldots,q_{i_m}$ such that:
    $$\{q_s,q_{i_1}\},\{q_{i_1},q_{i_2}\},\ldots,\{q_{i_{m-1}},q_{i_{m}}\},\{q_{i_{m}},q_k\} \ \text{ with } \ \{q_{i_t},q_{i_{t+1}}\} = \{q_i,q_j\}\cdot u_t \ \text{ for some } \ u_t\in\Sigma^*$$
    Being $V$ invariant under the action of $\A$, this gives us that:
    $$\pm(q_s- q_{i_1}) = \pm((q_s-q_1) - (q_{i_1}-q_1)), \ldots, \pm(q_{i_m}- q_k) = \pm((q_{i_m}-q_1) - (q_k-q_1)) \in V$$
    and thus, by summing everything (up to adjusting the coefficients) we obtain $(q_s-q_1)-(q_k-q_1)\in V$. This clearly proves that $V=\omega^\perp$, giving that $\A$ has to be irreducible.
\end{proof}

One might expect that, for a simple synchronizing automaton, weak defectivity is equivalent to irreducibility. However, the following proposition gives another sufficient condition (in the circular case) for the irreducibility, and in what follows we are able to use it to construct an infinite family of counterexamples (i.e. non-weakly defective, irreducible automata).

\begin{teo}\label{teo: contr irr non wd}
    Let $\A$ be a contracting automaton and assume that there is a word $u\in\Sigma^*$ such that $\ker(u) = \{q_i,\ Q\setminus\{q_i\}\}$ for some $q_i\in Q$. Then we have that $\A$ is irreducible. 
\end{teo}

\begin{proof}
    Let $V\subseteq \mathbb{C}^{n-1}$ be an invariant subspace under the action of $\A^*$ and $v=(v_2,\ldots,v_n)\in V$ a non-zero vector. Observe that, applying a suitable power of $a$ on $v$, we can assume that $v_i\neq 0$. Observe now that, for the $u\in\Sigma^*$ given by the hypothesis, we have: $Q\cdot u = \{q_j,q_k\}$ and by applying the contracting property together with a suitable power of $a$, we obtain an $s\in\Sigma^*$ such that $Q\cdot us = \{q_1,q_2\}$. In the coordinates of $\mathbb{C}^n$ we may write $v= (-\sum_iv_i)q_1 + v_2q_2 + \ldots + v_nq_n$ and observe that $v\cdot u$ has to be non zero by means of the hypothesis on $u$: indeed, we have that 
    $$v\cdot u= (-\sum_iv_i)q_1\cdot u + v_2q_2\cdot u + \ldots + v_nq_n\cdot u$$
    and by rewriting $v\cdot u$ in the basis  of $\mathbb{C}^n$, we must have that the coordinate associated to $q_i\cdot u$ has to be non-zero, hence $v\cdot u$ is non zero as an element of $\omega^\perp$.
    Thus we must have in standard basis $v\cdot us = v_i(-1,1,,0,\ldots,0)$.
    From this it follows that the vector $e=(-1,1,,0,\ldots,0) \in V$. Hence $f_{e} = x-1$ and we can conclude by means of Lemma \ref{lemma: span gcd}.
\end{proof}

As mentioned above, we conclude this section with the following example, providing an infinite class of irreducible, non-weakly defective automata:

\begin{ese}[Irreducible, non-weakly defective automata]\label{ese: irr non-w.-con.}
    Let $\A$ be the following automaton:
    
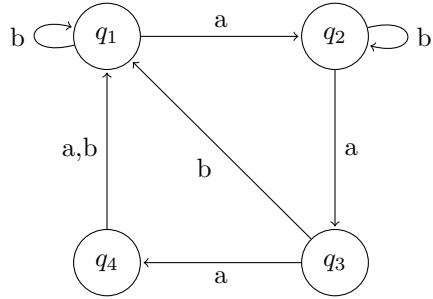
\begin{center}\begin{tikzpicture}[shorten >=1pt,node distance=3cm,on grid,auto] 
        \node[state] (q_1)   {$q_1$}; 
        \node[state] (q_2) [right=of q_1] {$q_2$}; 
        \node[state] (q_3) [below=of q_2] {$q_3$}; 
        \node[state] (q_4) [below=of q_1] {$q_4$};
        \path[->] 
            (q_1) edge node {a} (q_2)
                  edge [loop left] node {b}  ()
            (q_2) edge node {a} (q_3)
                  edge [loop right] node {b} ()
            (q_3) edge node {a} (q_4)
                  edge node {b} (q_1)
            (q_4) edge node {a,b} (q_1);
    \end{tikzpicture}\end{center}
    Observe that it is synchronizing by means of the word $ba^2b\in\Syn(\A)$. Observe that it is contracting: the only two pairs of states of distance $>1$ are $\{q_1,q_3\}$ and $\{q_2,q_4\}$ and $d(q_1\cdot ab,q_3\cdot ab) = 1$, $d(q_2\cdot b,q_4\cdot b)=1$ which proves the contracting property for $\A$. Observe now that $b\in\Sigma$ has rank equal to $2$ with $q_1\cdot b = q_2\cdot b = q_3\cdot b = q_1$ and $q_4\cdot b = q_4$, proving the hypothesis of the previous theorem, showing that $\A$ is irreducible. The construction can be generalized as follows: let $n\in\mathbb{Z}_{\geq 5}$, $\A_n$ be the automaton having $Q=\{q_1,\ldots,q_n\}$ as set of $n-$states, $\Sigma=\{a,b\}$ as alphabet with $a$ the circulating letter (as usual, such that $q_i\cdot a=q_{i+1}$) and $b$ such that $q_n\cdot b = q_{n-1}\cdot b = q_1$ and $q_i\cdot b = q_i$ whenever $i\neq n,n-1$; it is clear that $\A_n$ is not weakly defective. Let us show that $\A_n$ is contracting: let $q_i,q_j\in Q$ with $d(q_i,q_j)>1$ and $i<j$. By means of Remark \ref{oss: d = min}, we know that $d(q_i,q_j)=\min\{j-i,n-j+i\}$: let us consider the two cases
    \begin{itemize}
        \item if $d(q_i,q_j)=j-i$, observe that:
        \[\begin{split}
            1\leq d(q_i\cdot a^{n-i}b,q_j\cdot a^{n-i}b) &= d(q_1,q_j\cdot a^{n-i}) \\
            & =\min\{n-j+i-1,j-i-1\} \\
            & = j - i - 1 \\
            & < d(q_i,q_j)
        \end{split}\]
        which concludes.
        \item if $d(q_i,q_j)=n-j+i$, observe that:
        \[\begin{split}
            1\leq d(q_i\cdot a^{n-j-1}b,q_j\cdot a^{n-j-1}b) &= d(q_i\cdot a^{n-j-1},q_1) \\
            & =\min\{n-j-2+i,j-i+2\} \\
            & = n-j-2+i \\
            & < d(q_i,q_j)
        \end{split}\]
        which again concludes.
    \end{itemize}
    Let us now prove that $\A_n$ is synchronizing and satisfies the hypothesis of Theorem \ref{teo: contr irr non wd}: we clearly have that $(ba^{n-2})^{\lfloor n/2\rfloor}b$ is synchronizing, and $q_{n-2}\cdot(ba^{n-2})^{\lfloor n/2\rfloor} \neq q_i \cdot (ba^{n-2})^{\lfloor n/2\rfloor}$ for every $i\neq n-2$. This shows that $\A_n$ has to be irreducible.
\end{ese}

\section{Conclusion and open problems}
Regarding irreducibility, the only concrete results we have obtained so far rely on the contracting property. One might suspect that the class of irreducible automata is contained within the class of contracting automata, but the following result disproves this intuition.

\begin{prop}
    The automaton given in Example \ref{ese: weakly-con non-con} is irreducible.
\end{prop}

\begin{proof}
    Let $V\subseteq\mathbb{C}^7$ be invariant under the action of $\A^*$ and $v=(v_2,\ldots,v_8)\in V\setminus\{0\}$. Observe that:
    $$v\cdot b=(0,-v_2-v_4-v_8,0,0,v_2+v_4+v_8,0,0)$$
    and assume that $v_2+v_4+v_8 \neq 0$: this gives us that $(0,-1,0,0,1,0,0)\in V$. Observe that here $f_v=x^4-x = x(x^3-1)$ and it is straightforward to check that $\gcd(f_v,x^8-1) = x-1$, showing us that $V=\mathbb{C}^{n-1}$ by Lemma~\ref{lemma: span gcd}. Assume now that $v_2+v_4+v_8 = 0$: if $\exists k\geq 1$ such that $v\cdot a^kb \neq 0$, we can conclude as if $v_2+v_4+v_8 \neq 0$. Otherwise, assume that $v\cdot a^kb = 0$ for any $k\geq 0$: we have the following system:
    \[\begin{cases}
        v_2 + v_4 + v_8 = 0 \\
        v_1 + v_3 + v_5 = 0 \\
        v_2 + v_4 + v_6 = 0 \\
        v_3 + v_5 + v_7 = 0 \\
        v_4 + v_6 + v_8 = 0 \\
        v_1 + v_5 + v_7 = 0 \\
        v_2 + v_6 + v_8 = 0 \\
        v_1 + v_3 + v_7 = 0 \\
    \end{cases}\]
    where $v_1 = -\sum_{i}v_i$. From this system it follows thanks to some easy calculations that $v=0$ and thus we have a contradiction.
\end{proof}

It is important to notice that all the examples we manage to give of irreducible automata have former-rank $2$, that is, they all admit a word of rank $2$. This brings us to state the following conjecture:

\begin{conj}
    Let $\A$ be an irreducible automaton. Then $\exists u\in\Sigma^*$ with $\text{rk}(u)=2$.
\end{conj}

Still, in Example \ref{ese:simp non irr} we give a full class of simple, non-irreducible, former-rank $2$ automata. These factors combined suggest that we may have to consider other properties in order to find a full characterization of this class.
A second limitation of our current theory pertains to the reducibility results and examples we present: in every case, the action of the reducible automaton leaves only $1$-dimensional subspaces of $\omega^\perp\cong\mathbb{C}^{n-1}$ invariant. This leads us to state the following open problem:

\begin{opbm}
    Is there any example of reducible (circular) automaton $\A$ such that for all invariant subspaces $V\subseteq \mathbb{C}^{n-1}$ we have $\dim V\ge 2$? 
\end{opbm}

Observe that, if the answer to this question is negative (i.e., if every reducible automaton is invariant on a $1$-dimensional subspace), Theorem \ref{teo: bound red 1dim} would immediately yield a strict upper bound on the ranks of the letters of the alphabet. This would suggest that such automata (at least in the circular case) must be \emph{fastly synchronizing}. We conjecture that the aforementioned rank result can be straightforwardly generalized to the non-circular setting. \\
We conclude by analyzing Conjecture \ref{conj: extr are irred}. Observe that, by leveraging Theorem \ref{teo: Rystov}, we could address this problem by proving that every extremal automaton is simple and weakly defective (a property true for all known extremal automata). Furthermore, if the \v{C}ern\'y conjecture holds, it is highly plausible that the length of the shortest synchronizing word is inversely proportional to the defectivity of the alphabet's letters. Note that the only irreducible but non-weakly defective example we can provide (namely Example \ref{ese: irr non-w.-con.}) is far from being extremal. Therefore, we can combine \cite[Conjecture 7.21]{RodVen} with the following hypothesis (if proven) to obtain a solution to Conjecture \ref{conj: extr are irred}:
\begin{conj}
    Let $\A$ be an extremal automaton. Then $\A$ is weakly defective.
\end{conj}

\section*{Acknowledgments}

I am grateful to the anonymous referee for its careful reading and valuable feedback, which significantly improved the clarity of the paper. I would also like to thank Prof. Emanuele Rodaro and Mr. Gabriele Reali for their fruitful suggestions and the thoughtful conversations regarding the manuscript.

\printbibliography[title={References}]
	
\end{document}